\newif\ifconf
\newif\ifconf
\newif\ifconf
\conftrue
\ifconf
 \documentclass[letterpaper,conference,10pt]{ieeeconf}

 \newcommand\IEEEQED\QED
 
\else
 \documentclass[letterpaper,10pt,journal,twoside]{IEEEtran}
\fi

\usepackage{amsmath,mathtools,amssymb}
\usepackage{balance}
\newcounter{cAss}
\newcounter{cAssSaved}
\newcommand{\Ass}[1]{\ensuremath{\boldsymbol{\mathcal A}_{\mathbf{#1}}}}

\newlength\asswidth

\usepackage{amsfonts}
\usepackage{graphicx}
\usepackage{color}
\usepackage{float}
\usepackage{tabularx}
\usepackage[english]{babel}
\usepackage{gensymb}
\usepackage{amssymb}
\usepackage{array}
\usepackage{amsmath,mathtools}

\usepackage{graphicx}
\usepackage{epsfig}
\usepackage{subfigure}

\usepackage{multirow}
\usepackage{longtable}
\usepackage{booktabs}
\usepackage{tikz,pgfplots,pgfmath} 
\usetikzlibrary{math,positioning,intersections,shapes.misc,shapes.geometric,fit,arrows,arrows.meta,decorations.pathmorphing,decorations.pathreplacing,decorations.shapes,decorations.markings,patterns,calc,shadings}

\newtheorem{theorem}{Theorem}

\newtheorem{proposition}{Proposition}
\newtheorem{remark}{Remark}

\newtheorem{assumption}{Assumption}
\newtheorem{definition}{Definition}
\newcommand{\R}{\mathbb{R}}

\makeatletter
\newcommand{\@endgadget}[1]{%
  {\unskip\nobreak\hfil\penalty50\hskip1em\hbox{}\nobreak\hfil#1%
  \parfillskip=0pt\finalhyphendemerits=0\par}%
}

\newcommand{\@Endofsymbol}{$\triangledown$}   

\newcommand{\Endofremark}{\@endgadget{\@Endofsymbol}}
\makeatother

\title{\LARGE \bf
 Robust Safety-Critical Control of Integrator Chains with Mismatched Perturbations via Linear Time-Varying Feedback
}

\author{Imtiaz Ur Rehman$^{1,2}$, Moussa Labbadi$^{3}$, Member, IEEE, Amine Abadi$^{1}$ and Lew Lew Yan Voon$^{1}$
\thanks{*We thank the French government for the Plan France Relance initiative which provided fundings via the European Union under contract ANR-21-PRRD-0047-01. We are also grateful to the Doctoral School and the French Ministry of Research for the PhD MENRT scholarship.
}
\thanks{$^{1}$  I.U. Rehman, A. Abadi and L.L.Y. Voon are with  ImViA Laboratory EA 7535, Universite Bourgogne Europe, Hub\&Go - 72, Rue Jean Jaurès, 71200 Le Creusot, France,
        {\tt\small Imtiaz-Ur.Rehman@ube.fr; Amine.Abadi@ube.fr; lew.lew-yan-voon@ube.fr }}%
\thanks{$^{2}$ I.U. Rehman is with COMSATS University Islamabad, Islamabad, Pakistan
        {\tt\small imtiaz.rehman@comsats.edu.pk}}%
\thanks{$^{3}$ M. Labbadi is with Aix-Marseille University, LIS UMR CNRS 7020, Marseille, France, 
        {\tt\small moussa.labbadi@lis-lab.fr}}%
}

\IEEEoverridecommandlockouts

\begin{document}

\maketitle
\thispagestyle{empty}

\begin{abstract}
In this paper, we propose a novel safety-critical control framework for a chain of integrators subject to both matched and mismatched perturbations. The core of our approach is a linear, time-varying state-feedback design that simultaneously enforces stability and safety constraints. By integrating backstepping techniques with a quadratic programming (QP) formulation, we develop a systematic procedure to guarantee safety under time-varying gains. We provide rigorous theoretical guarantees for the double integrator case, both in the presence and absence of perturbations, and outline general proofs for extending the methodology to higher-order chains of integrators. This proposed framework thus bridges robustness and safety-critical performance, while overcoming the limitations of existing prescribed-time approaches.
\end{abstract}
\begin{keywords}
Safety-critical control, Linear time-varying feedback, Control barrier functions, Robust control, Matched and mismatched
perturbations.
\end{keywords}

\section{INTRODUCTION}
Safety is the primary requirement in modern autonomous/non-autonomous systems, as malfunctions may result in detrimental or even deadly consequences, including casualties or significant damage. For example, surgical robots must protect patients undergoing surgical operations~\cite{fan2024learn}, robotic manipulators need to work without hitting obstacles~\cite{ferraguti2020safety}, and autonomous vehicles are required to prevent traffic collisions~\cite{chen2017obstacle}. These challenges emphasize the necessity of safety-critical control techniques that can guarantee trustworthy and reliable performance in practical settings.

Control barrier functions (CBFs) have proven effective in ensuring the safety of dynamical systems, including robotic and autonomous systems~\cite{ames2016control,cortez2019control}. Traditional CBF formulations may be vulnerable to external perturbations or uncertainties, highlighting the importance of evaluating safety under such conditions~\cite{xu2015robustness}. 
Several robust control barrier function (RCBF) formulations have been proposed in the literature to address this limitation, leveraging different robustness mechanisms such as input-to-state stability margins, disturbance attenuation, and uncertainty compensation~\cite{jankovic2018robust,nguyen2021robust,kolathaya2018input,buch2021robust,garg2021robust}.
However, the authors used a constant bound, which limits flexibility, leading to conservatism and reduced performance. Although preserving safety is paramount under disturbances, it is also important to avoid overly conservative behavior, which might degrade performance. 

Recent research has developed time-varying control techniques to achieve both prescribed-time convergence and safety guarantees in safety-critical systems~\cite{abel2023prescribed,abel2022prescribed,nasab2025safe}. Although these algorithms deliver good convergence performance, they are often constrained by two key shortcomings: adequate robustness against disturbances and uncertainties, and the emergence of singularities as the time variable approaches the terminal time ($T$)~\cite{aldana2023inherent}. These drawbacks limit its practical applicability, particularly in unpredictable or rapidly varying environments. 
Alternatively, recent advances have introduced time-varying functions to achieve hyperexponential stability, ensuring fast convergence rates, inherent robustness to matched/ mismatched disturbances, and avoidance of singularities, all without the need for an explicit disturbance observer~\cite{labbadi2024hyperexponential}. Inspired by this concept, we include time varying gains straight into the safety design, allowing the barrier dynamics to adapt with disturbance-robustness attributes. This integration ensures constraint compliance despite both mismatched and matched perturbations.

Our key contribution is the augmentation of the backstepping-based RCBF framework through the introduction of a novel class of growing functions defined on $\mathbb{R}^+$, which rigorously guarantees hyperexponential stability while substantially enhancing transient performance. Unlike existing safe backstepping approaches~\cite{taylor2022safe,cohen2024safety,cohen2024constructive} that rely on disturbance-free dynamics, our methodology explicitly incorporates both mismatched and matched perturbations, thereby achieving robust constraint satisfaction with markedly reduced conservatism. Importantly, our framework dispenses with the need for disturbance observers, simplifying implementation and eliminating the computational load, design complexity, and tuning requirements associated with observer-based schemes~\cite{sun2024safety,wang2024safety,dacs2022robust,wang2024immersion,alan2022disturbance,agrawal2022safe}. This renders our approach both practically efficient and theoretically robust, bridging a significant gap between safety-critical control and real-world perturbation handling.


The paper is organized as follows. Section II covers preliminaries on CBF and problem formulation. Main results are presented in Section III. Section IV highlights the relevant simulations. Finally, some closing remarks are provided in Section V.






\textit{Notations}: This paper will utilize the notation listed below. Assume $\R^{+}$:=[0,$\infty$), and $||.||$ corresponds to a vector's Euclidean norm. A function $\alpha$ : $\R^{+}$ $\rightarrow$ $\R^{+}$ is a class $\kappa$ provided $\alpha(0)$=0, strictly increasing and continuous; if $\alpha(\infty)=\infty$ and is of class $\kappa$ then function $\alpha$ is a class $\kappa^{e}_{\infty}$.
\section{PRELIMINARIES AND PROBLEM FORMULATION}
Consider the nth order chain of integrator system and its disturbed version:
\begin{equation}\label{eq:1}
    \dot{x} = Ax + bu,
\end{equation}
\begin{equation}\label{eq:1_d}
    \dot{x} = Ax + bu+d(t),
\end{equation}
\[
A = \begin{bmatrix}
0_{(n-1) \times 1} & I_{n-1} \\
0 & 0_{1 \times (n-1)}
\end{bmatrix}, \quad 
b = 
\begin{bmatrix}
0_{(n-1) \times 1} \\
1
\end{bmatrix},
\]
where $x \in \mathbb{R}^{n}$ is the state vector, $u\in \mathbb{R}$ is the control input, and the disturbance vector is given by \( d(t) = [d_1(t) \ldots d_n(t)]^\top \in \mathbb{R}^n \) representing the mismatched and matched disturbances.

\begin{assumption}\label{eq:ass_1}
    The disturbances are bounded and
satisfies $\|d(t)\|\leq \theta$, $\theta>0$ is a known constant.
\end{assumption}

We introduce a few common concepts, inspired by~\cite{kim2025robust,ames2016control}. For a given set $\mathcal{\bar{C}} \subset \mathbb{R}^{n}$, we call $\mathcal{\bar{C}}$ forward invariant if for every $x(t_0) \in \mathcal{\bar{C}}$, it holds that $x(t) \in \mathcal{\bar{C}}$  $\forall t \geq t_0$, where $t_0\geq 0$ is the initialization time. System (\ref{eq:1}) is considered safe on the set $\mathcal{\bar{C}} \subset \mathbb{R}^{n}$ if and only if it is forward invariant.  
The safe set $\mathcal{S} \subset \mathbb{R}^{n}$ is defined as follows: 
\begin{align}\label{eq:2}
\mathcal{S} &= \{x\in \mathbb{R}^{n} : h(x) \geq 0\}, 
\end{align}
where, $h: \mathbb{R}^{n} \rightarrow \mathbb{R}$ is a continuously differentiable function.

\begin{definition}\label{def:1}
 ~\cite{kim2025robust}: For the system (\ref{eq:1}), a continuously differentiable function $h(x)$ is a CBF on the set $\mathcal{S}$ given in (\ref{eq:2}) such that $\forall x \in \mathcal{S}$:
\begin{equation}\label{eq:3} 
   \sup_{u \in U} [{L_{f}h(x)} + L_{g}h(x)u]  \geq - \Gamma(h(x)),   
\end{equation}
where $\Gamma \in \kappa^{e}_{\infty}$, $L_f h(x) = \left( \frac{\partial h}{\partial x} \right) Ax$ and $L_g h(x) = \left( \frac{\partial h}{\partial x} \right) b$.
\end{definition}
Here we introduce a new definition when system is subjected to external disturbances as defined in (\ref{eq:1_d}). 
\begin{definition}\label{def:2}
 ~\cite{zhao2020adaptive}: For the system (\ref{eq:1_d}), a continuously differentiable function $h(x)$ is a robust CBF (RCBF) on the set $\mathcal{S}$ given in (\ref{eq:2}) such that $\forall x \in \mathcal{S}$ and $\forall t \geq 0$ :
\begin{equation}\label{eq:5} 
   \sup_{u \in U} [{L_{f}h(x)} + L_{g}h(x)u - \|h_x(x)\| \theta ]  \geq - \Gamma(h(x)),   
\end{equation}
where  $h_x(x) = \frac{\partial h(x)}{\partial x}$. A fundamental problem in integrating CBF based approaches with backstepping
is the need for smooth control laws that can handle the successive
differentiations required at each stage of the backstepping algorithm.
This requirement is incompatible with the nonsmooth nature of conventional robust CBF constraints, like the one specified in (\ref{eq:5}), which poses an obstacle when computing multiple derivatives, often required for high relative degree situations~\cite{taylor2022safe}. We present the subsequent revised definition of RCBF by observing the subsequent upper bound of the non-smooth component using a smooth function for $\mu>0$: 
\begin{align}\label{eq:5_a}   
 \frac{1}{4\mu}\,\|h_x(x)\|^2 + \mu\theta^2 \geq \|h_x(x)\| \theta.
\end{align}
\end{definition}
\begin{definition}\label{def:3}
   For the system (\ref{eq:1_d}), a function $h(x):\mathbb R^n\to\mathbb R$ belongs to $\mathcal C^{n}$ is a smooth RCBF (SRCBF) on the set $\mathcal{S}$ given in (\ref{eq:2}) such that $\forall x \in \mathcal{S}$, for $\mu>0$ and $\forall t \geq 0$:
\begin{equation}\label{eq:5_b} 
   \sup_{u \in U} [{L_{f}h(x)} + L_{g}h(x)u -  \frac{1}{4\mu}\,\|h_x(x)\|^2 - \mu\theta^2 ]  \geq - \Gamma(h(x)).   
\end{equation}  
\end{definition}

 This study aims to design a safe time-varying control law for a chain of integrators, in the presence of both matched and mismatched perturbations, by employing linear time-varying feedback.

\section{Main results}
\subsection{Double integrator case without disturbance}
In this section, we first propose a control architecture for robust safety-critical control to obtain the safety objective for the double integrator system ($n=2$). We then extended this idea to a generalized n-th order chain of integrators.
Conventional CBF approaches are inadequate for systems with a high relative degree. This research leverages the backstepping approach from~\cite{krstic2006nonovershooting} to establish a new, relative-degree-one CBF and designing the corresponding controller using a Quadratic Program. To proceed with the design algorithm consider again system (\ref{eq:1}) with a desired candidate CBF $h_1(x)$
of relative degree $n > 1$, under the subsequent assumption:
\begin{assumption}\label{eq:ass_2}
The function $h_1(x):\mathbb R^2\to\mathbb R$ belongs to $\mathcal C^{2}$ and satisfies
\begin{equation}
    \frac{\partial h_1(x)}{\partial x} \neq 0, \quad \forall x \in \{x \in \mathbb R^2 \mid h_1(x) \geq 0  \}. 
\end{equation}
\end{assumption}
Computing the time derivative of $h_1(x)$ results in
\begin{align}\label{eq:8}
    \dot{h}_1(x) &= L_f h_1(x), \nonumber\\
                 &= -\bar{\varrho}_1 \Upsilon(t) h_1(x) + \underbrace{\bar{\varrho}_1 \Upsilon(t) h_1(x) + L_f h_1(x)}_{h_2(x)},
\end{align}
where $\bar{\varrho}_1$ is positive design parameter to be determined, and $\Upsilon(t)=1+t$ is a
strictly increasing time varying function. 
\begin{remark}\label{rm:3}
Any continuous, strictly increasing function 
$\Upsilon : \mathbb{R}^{+} \to \mathbb{R}^{+}$ with $\Upsilon(0) > 0$ 
and an unbounded integral can be utilized in \eqref{eq:8}. 
For instance, one may select linear, polynomial, or exponential forms, such as
\begin{gather*}
\Upsilon(t) = 1+t, \quad 
\Upsilon(t) = (1+t)^{p}, \; p>0, \quad \\
\Upsilon(t) = ae^{\alpha t}, \; a>0, \alpha>0 .
\end{gather*}\Endofremark
\end{remark}
\begin{remark}\label{rm:4}
The prescribed-time form $\Upsilon(t)=\frac{\Upsilon_0}{T-t}$, with $\Upsilon_0>0$, $T>0$, $t\in[0,T)$, can enforce finite-time safety but suffers from singularity at $t=T$, which prohibits its usage for persistent safety applications, and also lack robustness owing to noise/disturbance amplification. 
On the other hand, the functions in Remark~\ref{rm:3} circumvent singularity and are recommended in persistent safety-critical applications. 
\Endofremark
\end{remark}
Condition (\ref{eq:3}) holds for $h_1(x)$, provided that $h_2(x) \geq 0$. This results in a new candidate CBF defined as
\begin{align}\label{eq:9}
    {h}_2(x) &= \bar{\varrho}_1 \Upsilon(t) h_1(x) + L_f h_1(x).
\end{align}
However, it is necessary to ensure that the initial condition lies within the safe set defined by $h_2(x)$ $(i.e., h_2(x_0)>0)$. To resolve this, we impose the following assumption:
\begin{assumption}\label{eq:ass_3}
  $ h_1(x_0)>0$.
\end{assumption}

Typically, CBF-based safety frameworks assume $h_{1}(x_{0}) \geq 0$. Nevertheless, we do not consider the case where the initial condition lies exactly on the boundary of the safe set $\mathcal{S}$ (i.e., $h_{1}(x_{0}) = 0$), as enforcing safety from such points may result in control feasibility concerns. Under Assumption~\ref{eq:ass_3}, we select
\begin{equation}\label{eq:10}
   \bar{\varrho}_1 > \max\!\left\{ 0,\; -\,\frac{L_f h_1(x_0)}{\Upsilon(t_0)\,h_1(x_0)} \right\},
\end{equation}
this ensures $h_2(x_0)>0$ in accordance with (\ref{eq:9}).
A backstepping transformation is used to determine the desired CBF, which is expressed as follows:
\begin{align}\label{eq:11}
    h_1(x) &= h_1(x), \nonumber\\
    h_2(x) &= \bar{\varrho}_1 \Upsilon(t) h_1(x) + L_f h_1(x),
\end{align}
where $\bar{\varrho}_1$ is detailed in (\ref{eq:10}).
Finally, given the resulting CBF \(h_2(x)\), which has a relative degree of one, a safety-critical control law is explicitly formulated as follows in (\ref{eq:13}). This law modifies the nominal control $u_{\mathrm{no}}$ whenever the states approach the boundary of \(\mathcal{S}\). Consequently, the final control $u$ enforces the safety constraint, preventing violations and ensuring that the set \(\mathcal{S}\) remains forward invariant:
\begin{align}\label{eq:13}  
u &= \arg\min_{u} \; \|u - u_{\mathrm{no}}\|^2 \nonumber \\ 
\text{s.t. } &\underbrace{\ L_f h_2(x) + L_g h_2(x)\, u}_{\dot{h}_2(x,u)} \geq -\bar{\varrho}_2 \, \Upsilon(t)\, h_2(x).
\end{align}
whereas $\bar{\varrho}_2>0$. The explicit solution of this QP is obtained by employing the Karush-Kuhn-Tucker (KKT) optimality conditions, resulting as follows~\cite{boyd2004convex}:
\begin{align}\label{eq:18_a} 
u = \begin{cases} 
u_{no}, & \zeta(x, u_{no}) \geq 0, \\ 
u_{no} - (L_g h_2(x))^\top \frac{\zeta(x, u_{no})}{\|L_g h_2(x)\|^2}, & \text{otherwise},
\end{cases}
\end{align}
where
\begin{align}\label{eq:19} 
\zeta(x, u_{no}) &:= \dot{h}_2(x,u_{no}) + \bar{\varrho}_2 \Upsilon(t) h_2(x).
\end{align}
\begin{remark}
    In contrast to~\cite{abel2023prescribed}, our work addresses the problem of persistent safety, which is particularly necessary when the system must preserve safety eternally, for example in collision avoidance for driverless cars. The method in~\cite{abel2023prescribed}, on the other hand, is designed for systems that require temporary safety, like a human receiving an object from a robotic arm at a predetermined time, or any operation where safety enforcement is no longer necessary after a certain period. Rather than employing a blow-up function, in which the gains diverge to infinity and cause a singularity, we adopt a strictly increasing yet bounded function. This ensures that the gains grow monotonically over time while remaining finite, thereby preventing singularities. 
      \Endofremark
\end{remark}

\begin{proposition}\label{pro:1}
Consider the states of the unperturbed system (\ref{eq:1}) for $(n=2)$, which are initially safe, (i.e., $h_1(x_0) > 0$), and Assumption~\ref{eq:ass_2} hold. Then, the control law (\ref{eq:18_a}), designed from the CBFs (\ref{eq:11}) with the initial gain (\ref{eq:10}), ensures that $h_1(x(t)\geq 0)$ $\forall t \geq t_0$.
\end{proposition}
\begin{proof}
    The reasoning follows analogously to the proof of Theorem 1 in~\cite{abel2023prescribed}. Given $h_1(x_0) > 0$, the initial gain (\ref{eq:10}) is designed to guarantee $h_2(x_0) > 0$. Since we know $h_1(x_0) > 0$, the gain
choice $\bar{\varrho}_1$ in (\ref{eq:10}) ensures $h_2(x_0) > 0$. Differentiate $h_{1}(x)$ and $h_{2}(x)$ gives:
\begin{align}\label{eq:20_a} 
    \dot h_{1}(x) &= L_{f} h_{1}(x), \nonumber\\
    &\geq -\bar{\varrho}_1 \Upsilon(t)h_{1}(x) + \underbrace{\bar{\varrho}_1 \Upsilon(t) h_{1}(x) + L_{f} h_{1}(x) }_{h_2(x)}, \nonumber\\ 
     \dot h_{2}(x) &= L_{f} h_{2}(x) + L_{g} h_{2}(x)u,\nonumber\\
  &\geq -\bar{\varrho}_2 \Upsilon(t)\, h_{2}(x).
\end{align}
By leveraging the variation of constants formula and the comparison lemma to (\ref{eq:20_a}) yields
\begin{align}\label{eq:25_a}
  h_{1}(x(t)) &\ge h_{1}(x_{0}) \exp\Big(-\bar{\varrho}_1 \int_{t_{0}}^{t} \Upsilon(s)\,ds\Big) \nonumber\\
&\quad + \int_{t_{0}}^{t} h_{2}(s) \exp\Big(-\bar{\varrho}_1 \int_{t_0}^{t} \Upsilon(\tau)\,d\tau\Big) ds.  
\end{align}
\begin{small}
\begin{align}\label{eq:26_a}
  h_{2}(x(t),t) &\ge h_{2}(x_{0},t_{0}) \exp\Big(-\bar{\varrho}_2 \int_{t_{0}}^{t} \Upsilon(s)\,ds\Big) > 0,\quad \forall t \ge t_{0}. 
\end{align}    
\end{small}
Substituting (\ref{eq:26_a}) into (\ref{eq:25_a}) gives
\begin{equation}
\label{eq:27_a}
\begin{aligned}
h_1(x(t)) &\ge h_1(x_0) 
\exp\Bigg[ -\bar{\varrho}_1 \Big( \frac{(t - t_0)^2}{2} + (t - t_0) \Big) \Bigg] \\
&> 0, \quad \forall t \ge t_0.
\end{aligned}
\end{equation}

this implies that
\begin{align}\label{eq:27_aa}
  h_{1}(x(t)) &>0, \quad \forall t\geq t_0 
\end{align} 
\end{proof}
\subsection{Double integrator case with disturbance}
The above QP formulation ensures safety in the disturbance-free case. We now extend the backstepping technique to systems subject to perturbations. For the control-affine system with disturbances defined in (\ref{eq:1_d}), we assume a desired candidate CBF $h_1(x)$ with relative degree two that also fulfills Assumption \ref{eq:ass_2}.

By adopting the SRCBF formulation, we obtain a smooth controller that
avoids the singularities associated with nonsmooth constraints. 
We start over by calculating the time derivative of $h_1(x)$ 
\begin{small}
\begin{align}\label{eq:14}
    \dot{h}_1(x) &= L_f h_1(x)+\frac{\partial h_1}{\partial x}d, \nonumber\\
                 &\geq -\varrho_1 \Upsilon(t) h_1(x) + \underbrace{\varrho_1 \Upsilon(t) h_1(x) + L_f h_1(x) - \Lambda_{1}(x)}_{h_2(x)}, 
\end{align}    
\end{small}
where $\Lambda_{1}(x)= \frac{1}{4\mu_1}\,\|\frac{\partial h_1}{\partial x}\|^2 + \mu_{1}\theta^2$
and $\mu_1>0$. Then, using the same concept and method as previously described, we determined a target CBF as follows:
\begin{align}\label{eq:15}
    h_1(x) &= h_1(x), \nonumber\\
    h_2(x) &= \varrho_1 \Upsilon(t) h_1(x) + L_f h_1(x)-\Lambda_{1}(x),
\end{align}
where
\begin{align}\label{eq:16}
   \varrho_1 > \max\!\left\{ 0,\; \frac{-L_f h_1(x_0)+\Lambda_{1}(x_0)}{\Upsilon(t_0)\,h_1(x_0)} \right\}.
\end{align}
Therefore, the QP problem is expressed as:
\begin{align}\label{eq:17}  
u &= \arg\min_{u} \; \|u - u_{\mathrm{no}}\|^2 \nonumber \\ 
\text{s.t. } &\dot{h}_2(x,u) - \Lambda_{2}(x) \geq -\varrho_2 \, \Upsilon(t)\, h_2(x),
\end{align}
whereas $\Lambda_{2}(x)= \frac{1}{4\mu_2}\,\|\frac{\partial h_2}{\partial x}\|^2 + \mu_{2}\theta^2$
and $\mu_2>0$. The Karush-Kuhn-Tucker (KKT) optimality conditions are applied to determine the explicit solution to this QP, which results in:
\begin{align}\label{eq:18} 
u = \begin{cases} 
u_{no}, & \zeta(x, u_{no}) \geq 0, \\ 
u_{no} - (L_g h_2(x))^\top \frac{\zeta(x, u_{no})}{\|L_g h_2(x)\|^2}, & \text{otherwise},
\end{cases}
\end{align}
where
\begin{align}\label{eq:19_1} 
\zeta(x, u_{no}) &:= \dot{h}_2(x,u_{no}) - \Lambda_{2}(x) + \varrho_2 \Upsilon(t) h_2(x).
\end{align}
\begin{remark}
    While the framework in~\cite{kim2025robust} provides an effective solution for state-dependent disturbances, such as unmodeled rotor dynamics or aerodynamic drag that vanish when the drone is hovering, our work targets continuous, time-varying, bounded external perturbations that are independent of the system state, such as wind gusts affecting the drone whether it is moving or hovering, inspired by~\cite{labbadi2024hyperexponential}. This distinction is critical because a controller designed to stabilize a hovering drone would be vulnerable to wind at that position if the disturbance model assumes it vanishes at equilibrium, potentially resulting in a safety breach. Our work addresses this gap by ensuring robustness against these common external disturbances. 
    \Endofremark
\end{remark}
{\color{blue}




}
\begin{theorem}\label{th:1}
Consider the states of the perturbed system (\ref{eq:1_d}) for $(n=2)$ is initially safe and away from the barrier (i.e., $h_1(x_0) > 0$), and Assumptions~\ref{eq:ass_1}-\ref{eq:ass_2} hold. Then, the control law (\ref{eq:18}), synthesized from the CBFs (\ref{eq:15}) with the initial gain (\ref{eq:16}), guarantees that $h_1(x(t)\geq 0)$ $\forall t \geq t_0$.
\end{theorem}
\begin{proof}
Given $h_1(x_0) > 0$, the initial gain (\ref{eq:16}) is designed to
guarantee $h_2(x_0) > 0$. Since we know $h_1(x_0) > 0$, the gain
choice $\varrho_1$ ensures $h_2(x_0) > 0$. Differentiate $h_{1}(x)$ along (\ref{eq:1_d}) yields:
\begin{align}\label{eq:20} 
    \dot h_{1}(x) &= L_{f} h_{1}(x) + \frac{\partial h_1}{\partial x}d.
\end{align}
applying Assumption~\ref{eq:ass_1} and pointwise disturbance bound gives $||\frac{\partial h_1}{\partial x}||\theta \ge -\Lambda_{1}(x)$
and therefore (\ref{eq:20}) can be expressed as
\begin{small}
 \begin{align}\label{eq:22}
 \dot h_{1}(x) &\ge -\varrho_1 \Upsilon(t)h_{1}(x) + \underbrace{\varrho_1 \Upsilon(t) h_{1}(x) + L_{f} h_{1}(x) - \Lambda_{1}(x)}_{h_2(x)}.   
\end{align}   
\end{small}
hence
\begin{align}\label{eq:23}
\dot h_{1}(x) &\ge -\varrho_1 \Upsilon(t) h_{1}(x) + h_{2}(x).
\end{align}
Differentiating $h_{2}(x)$ along (\ref{eq:1_d}), applying Assumption~\ref{eq:ass_1} and disturbance bound one can obtain:
\begin{align}\label{eq:24}
  \dot h_{2}(x)&= L_{f} h_{2}(x) + L_{g} h_{2}(x)u - \Lambda_{2}(x),\nonumber\\
  &\geq -\varrho_2 \Upsilon(t)\, h_{2}(x)
\end{align}
Employing the analytical framework of proposition~\ref{pro:1}, along with the variation of constants formula and the comparison lemma on (\ref{eq:23}) and (\ref{eq:24}) yield
\begin{align}\label{eq:27}
  h_{1}(x(t)) &\ge h_1(x_0) 
\exp\Bigg[ -\varrho_1 \Big( \frac{(t - t_0)^2}{2} + (t - t_0) \Big) \Bigg] \\ \nonumber
&> 0, \quad \forall t \ge t_0, 
\end{align}
this demonstrates that
\begin{align}\label{eq:27_b}
  h_{1}(x(t)) &>0, \quad \forall t\geq t_0 
\end{align} 
\end{proof}
\subsection{Generalized n-th order chain of integrators}
The backstepping procedure can be extended to an $n$-th order system. An important design consideration is the scaling of the time-varying function $\Upsilon(t)$ at each transformation stage. Reinforcing the impact of $\Upsilon(t)$ at each stage is an effective approach for ensuring robustness and elevating performance for higher-order systems. In addition to expediting convergence to the inside of the safety set for the nominal system, this method provides a more powerful, time-varying feedback mechanism that is crucial for controlling the propagation of perturbations along the integrator chain. To accomplish this, $\Upsilon(t)$ is raised to a power proportional to the step index $i$. Inspired by~\cite{abel2023prescribed} and~\cite{labbadi2024hyperexponential}, this methodology offers robust safety guarantees for systems of any relative degree through the inclusion of a substantial time-varying feedback gain.
\begin{assumption}\label{eq:ass_4}
The function $h_1(x):\mathbb R^n\to\mathbb R$ belongs to $\mathcal C^{n}$ and satisfies
\begin{equation}
    \frac{\partial h_1(x)}{\partial x} \neq 0, \quad \forall x \in \mathcal{S}. 
\end{equation}
\end{assumption}
Employing the backstepping transformation to the unperturbed system (\ref{eq:1}) proceeds as follows:
\begin{align}\label{eq:11_a}
    h_1(x) &= h_1(x), \nonumber\\
    h_i(x) &= \bar{\varrho}_{i-1} \Upsilon(t)^{\vartheta(i-1)} h_{i-1}(x) + L_f h_{i-1}(x),
\end{align}
for $i=\{2, \ldots, n\}$, where the aggressiveness of the variable gain scaling is controlled by the tuning factor $\vartheta \geq 1$. while, the gains $\bar{\varrho}_{i-1}$ are chosen to ensure that each $h_i(x, t_0)$ is initially positive: 
\begin{equation}\label{eq:10_a}
    \bar{\varrho}_{i-1} > \max\!\left\{ 0,\; -\,\frac{L_f h_{i-1}(x_0)}{\Upsilon(t_0)^{\vartheta(i-1)}\,h_{i-1}(x_0)} \right\}.
\end{equation}
Consequently, the QP problem is formulated as:
\begin{align}\label{eq:13_a}  
u &= \arg\min_{u} \; \|u - u_{\mathrm{no}}\|^2 \nonumber \\ 
\text{s.t. } & {\dot{h}_n(x,u)} \geq -\bar{\varrho}_{n} \, \Upsilon(t)^{\vartheta n}\, h_n(x),
\end{align}
where $\bar{\varrho}_{n}>0$. An explicit solution can also be obtained by using the KKT optimality conditions:
\begin{align}\label{eq:18_b} 
u = \begin{cases} 
u_{no}, & \zeta(x, u_{no}) \geq 0, \\ 
u_{no} - (L_g h_n(x))^\top \frac{\zeta(x, u_{no})}{\|L_g h_n(x)\|^2}, & \text{otherwise},
\end{cases}
\end{align}
wherein
\begin{align}\label{eq:19_a} 
\zeta(x, u_{no}) &:= {\dot{h}_n(x,u_{no})} +\bar{\varrho}_{n} \, \Upsilon(t)^{\vartheta n}\, h_n(x).
\end{align}
Similarly, for the fully perturbed system (\ref{eq:1_d}), the application of backstepping transformation gives:
\begin{small}
\begin{align}\label{eq:11_b}
    h_1(x) &= h_1(x), \nonumber\\
    h_i(x) &= \varrho_{i-1} \Upsilon(t)^{\vartheta(i-1)} h_{i-1}(x) + L_f h_{i-1}(x)-\Lambda_{i-1}(x),
\end{align}
\end{small}
for $i=\{2, \ldots, n\}$ and $\Lambda_{i-1}(x)= \frac{1}{4\mu_{i-1}}\,\|\frac{\partial h_{i-1}}{\partial x}\|^2 + \mu_{i-1}\theta^2$, while adjusting the initial gains to incorporate the perturbation bound:
\begin{equation}\label{eq:10_b}
    \varrho_{i-1} > \max\!\left\{ 0,\; \,\frac{-L_f h_{i-1}(x_0)+\Lambda_{i-1}(x_0)}{\Upsilon(t_0)^{\vartheta(i-1)}\,h_{i-1}(x_0)} \right\}.
\end{equation}
Resultantly, the QP problem and its explicit solution is expressed as:
\begin{align}\label{eq:13_b}  
u &= \arg\min_{u} \; \|u - u_{\mathrm{no}}\|^2 \nonumber \\ 
\text{s.t. } &{\dot{h}_n(x,u)} -\Lambda_{n}(x) \geq -\varrho_{n} \, \Upsilon(t)^{\vartheta n}\, h_n(x).
\end{align}
\begin{align}\label{eq:18_c} 
u = \begin{cases} 
u_{no}, & \zeta(x, u_{no}) \geq 0, \\ 
u_{no} - (L_g h_n(x))^\top \frac{\zeta(x, u_{no})}{\|L_g h_n(x)\|^2}, & \text{otherwise},
\end{cases}
\end{align}
where
\begin{align}\label{eq:19_d} 
\zeta(x, u_{no}) &:= {\dot{h}_n(x,u_{no})} -\Lambda_{n}(x) +\varrho_{n} \, \Upsilon(t)^{\vartheta n}\, h_n(x).
\end{align}
\begin{proposition}\label{pro:2}
If the system (\ref{eq:1}) is initially safe, i.e., $h_1(x_0) > 0$, and Assumption~\ref{eq:ass_4} satisfies, then the control law (\ref{eq:18_b}), synthesized employing the backstepping transformation (\ref{eq:11_a}) with the initial gain (\ref{eq:10_a}) and $\bar{\varrho}_n > 0$, ensures that $h_1(x(t))\geq 0$ $\forall t \in [t_0, \infty)$.
\end{proposition}
\begin{proof}
   Given $h_{i-1}(x_0) > 0$, the initial gain (\ref{eq:10_a}) is designed to guarantee $h_i(x_0) > 0$. Since we know $h_1(x_0) > 0$, the gain choice 
\begin{equation}\label{eq:10_aa}
   \bar{\varrho}_1 > \max\!\left\{ 0,\; -\,\frac{L_f h_1(x_0)}{\Upsilon(t_0)\,h_1(x_0)} \right\},
\end{equation}
ensures $h_2(x_0) > 0$. Thus, by induction, we can conclude $h_i(x_0) > 0$ for $i=\{2, \ldots, n\}$. Then, by differentiating the chain of CBFs (\ref{eq:11_a}) gives: 
\begin{align}\label{eq:28} 
\frac{d}{dt}h_k(x(t))&=-\bar{\varrho}_{k} \Upsilon(t)^{\vartheta k}h_{k}(x(t)) + h_{k+1}(x(t)),
\end{align}
\begin{align}\label{eq:29} 
\frac{d}{dt}h_n(x(t))&=L_f h_n(x(t))+L_g h_n(x(t))u \nonumber\\
                    &\geq  -\bar{\varrho}_{n} \, \Upsilon(t)^{\vartheta n}\, h_n(x(t)),
\end{align}
for $k=\{1, \ldots, n-1\}$.
Employing the  comparison lemma and the variation of constants formula, for $t \in [t_0, \infty)$, leads to
\begin{align}\label{eq:30}
  h_{k}(x(t)) &\ge h_{k}(x({t_0})) e^{-\bar{\varrho}_k \int_{t_{0}}^{t} \Upsilon(s)^{\vartheta k}\,ds} \nonumber\\
&\quad + \int_{t_{0}}^{t} h_{k+1}(s) e^{-\bar{\varrho}_k \int_{t_0}^{t} \Upsilon(\tau)^{\vartheta k}\,d\tau} ds,  
\end{align}
\begin{small}
\begin{align}\label{eq:31}
  h_{n}(x(t)) &\ge h_{n}(x(t_{0})) e^{-\bar{\varrho}_n \int_{t_{0}}^{t} \Upsilon(s)^{\vartheta n}\,ds}, 
\end{align}    
\end{small}
for $k=\{1, \ldots, n-1\}$. As established before, $h_n(x(t_0))>0$, which implies from (\ref{eq:31}) that $h_n(x(t))>0$ for all $t \in [t_0, \infty)$. Now, substituting (\ref{eq:31}) into (\ref{eq:30}) for $k=n-1$ results in
\begin{align}\label{eq:32}
  h_{n-1}(x(t)) &\ge \underbrace{h_{n-1}(x({t_0})) e^{-\bar{\varrho}_{n-1} A(t)}}_{>0} \nonumber\\
& + \underbrace{ h_n(x(t_0)) e^{-\bar{\varrho}_{n-1} A(t)} \int_{t_0}^{t} e^{\bar{\varrho}_{n-1} A(\tau) - \varrho_{n} B(\tau) } d\tau}_{ \geq 0 }, \nonumber \\
& \geq h_{n-1}(x({t_0})) e^{-\bar{\varrho}_{n-1} A(t)} > 0,
\end{align}
where $\bar{A}(t)=\int_{t_0}^{t} \Upsilon(s)^{\vartheta(n-1)}  ds$ and $\bar{B}(t)=\int_{t_0}^{t} \Upsilon(s)^{\vartheta n}  ds$. By applying backward strong induction with (\ref{eq:31}) and (\ref{eq:32}),
it can be demonstrated that
\begin{align}\label{eq:33}
  h_{1}(x(t)) &\ge h_{1}(x(t_{0})) e^{-\bar{\varrho}_1 \int_{t_{0}}^{t} \Upsilon(s)\,ds}, 
\end{align} 
which implies
\begin{align}\label{eq:34}
  h_{1}(x(t)) &\ge  0, \quad \forall t \in [t_0, \infty). 
\end{align} 
\end{proof}
\begin{theorem}\label{th:2}
   For the disturbed system (\ref{eq:1_d}), assume that the initial state is safe, i.e., $h_1(x_0) > 0$, and that it satisfies Assumptions~\ref{eq:ass_1} and~\ref{eq:ass_4}. Then, the control law (\ref{eq:18_c}), derived using the backstepping transformation (\ref{eq:11_b}) with the initial gain (\ref{eq:10_b}) and $\varrho_n > 0$, ensures that $h_1(x(t))\geq 0$ $\forall t \in [t_0, \infty)$. 
\end{theorem}
\begin{proof}
    Given $h_{i-1}(x_0) > 0$, the initial gain (\ref{eq:10_b}) is designed to ensure $h_i(x_0) > 0$. Since we know $h_1(x_0) > 0$, the gain choice
\begin{align}\label{eq:35}
     \varrho_1 > \max\!\left\{ 0,\; \frac{-L_f h_1(x_0)+\Lambda_{1}(x_0)}{\Upsilon(t_0)\,h_1(x_0)} \right\},
\end{align}
guarantees $h_2(x_0) > 0$. Hence, by induction it follows that $h_i(x_0) > 0$ for $i=\{2, \ldots, n\}$. Differentiating the chain of CBFs (\ref{eq:11_b}) yields: 
\begin{align}\label{eq:28_a} 
\frac{d}{dt}h_k(x(t))&=-\varrho_{k} \Upsilon(t)^{\vartheta k}h_{k}(x(t)) + h_{k+1}(x(t)),
\end{align}
\begin{align}\label{eq:29_a} 
\frac{d}{dt}h_n(x(t))&=L_f h_n(x(t))+L_g h_n(x(t))u-\Lambda_{n}(x) \nonumber\\
                    &\geq  -\varrho_{n} \, \Upsilon(t)^{\vartheta n}\, h_n(x(t)),
\end{align}
for $k=\{1, \ldots, n-1\}$. Using the variation of constants formula and the comparison lemma, for $t \in [t_0, \infty)$, one can get
\begin{align}\label{eq:30_a}
  h_{k}(x(t)) &\ge h_{k}(x({t_0})) e^{-\varrho_k \int_{t_{0}}^{t} \Upsilon(s)^{\vartheta k}\,ds} \nonumber\\
&\quad + \int_{t_{0}}^{t} h_{k+1}(s) e^{-\varrho_k \int_{t_0}^{t} \Upsilon(\tau)^{\vartheta k}\,d\tau} ds,  
\end{align}
\begin{small}
\begin{align}\label{eq:31_a}
  h_{n}(x(t)) &\ge h_{n}(x(t_{0})) e^{-\varrho_n \int_{t_{0}}^{t} \Upsilon(s)^{\vartheta n}\,ds}, 
\end{align}    
\end{small}
for $k=\{1, \ldots, n-1\}$. As previously stated, $h_n(x(t_0))>0$. Consequently, relation (\ref{eq:31_a}) ensures $h_n(x(t))>0$ for all $t \in [t_0, \infty)$. Now , substituting (\ref{eq:31_a}) into (\ref{eq:30_a}) for $k=n-1$ results in
\begin{align}\label{eq:32_a}
  h_{n-1}(x(t)) &\ge \underbrace{h_{n-1}(x({t_0})) e^{-\varrho_{n-1} A(t)}}_{>0} \nonumber\\
& + \underbrace{ h_n(x(t_0)) e^{-\varrho_{n-1} A(t)} \int_{t_0}^{t} e^{\varrho_{n-1} A(\tau) - \varrho_{n} B(\tau) } d\tau}_{ \geq 0 }, \nonumber \\
& \geq h_{n-1}(x({t_0})) e^{-\varrho_{n-1} A(t)} > 0,
\end{align}
where $A(t)=\int_{t_0}^{t} \Upsilon(s)^{\vartheta(n-1)}  ds$ and $B(t)=\int_{t_0}^{t} \Upsilon(s)^{\vartheta n}  ds$. By using backward strong induction with (\ref{eq:31_a}) and (\ref{eq:32_a}),
it follows that
\begin{align}\label{eq:33_a}
  h_{1}(x(t)) &\ge h_{1}(x(t_{0})) e^{-\varrho_1 \int_{t_{0}}^{t} \Upsilon(s)\,ds}, 
\end{align} 
hence
\begin{align}\label{eq:34_a}
  h_{1}(x(t)) &\ge  0, \quad \forall t \in [t_0, \infty). 
\end{align} 
\end{proof}




\section{Simulations}
The double integrator system, represented by $\dot{p} = v + d_{1}(t)$ and $\dot{v} = u + d_{2}(t)$, is used in the simulation to demonstrate the effectiveness of the proposed strategy. Both mismatched disturbances $d_{1}(t)$ and matched disturbances $d_{2}(t)$ are introduced. The position, velocity, disturbances, and control input are captured by $p$, $v$, ${d_{1}, d_{2}}$, and $u \in \mathbb{R}^2$, respectively. In addition, the goal is to drive the system to the desired position $p_d \in \mathbb{R}^2$ while subjected to perturbations, ensuring effective avoidance of an obstacle centered at $x_c \in \mathbb{R}^2$ with radius $r > 0$. Refraining Collision is recorded using the safe set $\mathcal{S}$ with: $h_{1}(p)=\frac{1}{2}(||p-x_{c}||^{2}_{2}-r^{2})$. The perturbation profiles are set as $d_1(t)= [0.1sin(2t)+0.02\mathrm{rnd}(1); 0.1cos(3t)+0.02\mathrm{rnd}(1)]$ and $d_2(t)= [0.15sin(t)+0.02\mathrm{rnd}(1); 0.15cos(2t)+0.02\mathrm{rnd}(1)]$. The parameter $\varrho_{1}=2.7$ is chosen such that $h_2(p_0)>0$. As demonstrated in Theorem~\ref{th:1} that ensuring $h_2(p_0)>0$ ensures system safety, while $\varrho_{2}=3$ is taken as the initial gain with a factor of smoothness $\mu_1=\mu_2=0.2$ and $\Upsilon(t)=1+t$. Moreover, the simulation involves two different controllers. The first one employs the standard CBF backstepping method, referred to as (SBCBF) from~\cite{krstic2006nonovershooting}, and the second one employs our proposed methodology. Furthermore, the nominal objective is to stabilize the system at $p_d$ to achieve this, we employ the nominal controller prescribed in~\cite{labbadi2024hyperexponential}. 

The corresponding system trajectories are illustrated in Fig.~\ref{fig:result1}. To demonstrate the original behavior without a safety filter, we also simulate the nominal controller. As anticipated, the nominal controller stabilizes the system without taking the safety into consideration. Fig.~\ref{fig:result2} shows the profile of the desired candidate CBF $h_{1}(p)$ and the control input for both controllers. In Fig.~\ref{fig:result1}, the SBCBF does not account for the perturbations, resulting in a collision and thus violating safety. In contrast, our proposed approach demonstrates robustness, adheres to the safety constraint, and entails minimal control effort.


\begin{figure}[]
    \centering
    \includegraphics[scale=0.58]{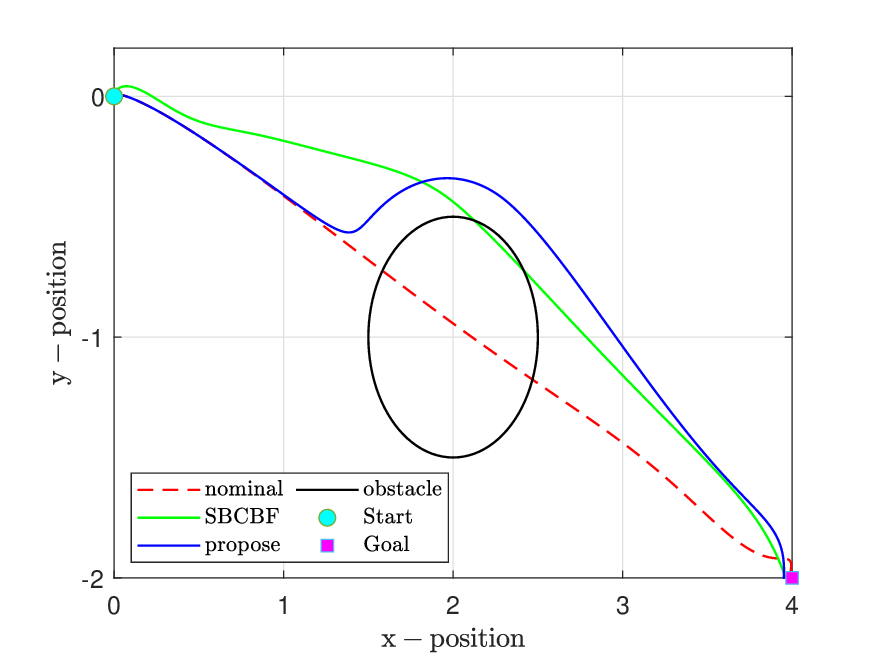}
    \caption{Double integrator system under nominal law, SBCBF, and the proposed safety filter. The nominal controller reaches the desired goal but intersects the obstacle. The SBCBF fails to prevent collision under perturbations, as indicated by the trajectory portion inside the obstacle. In contrast, the proposed safety filter accomplishes the objective by effectively guiding the system to evade the obstacle, thus preventing a collision.}
    \label{fig:result1}
\end{figure}
\begin{figure}[]
    \centering
    \includegraphics[scale=0.58]{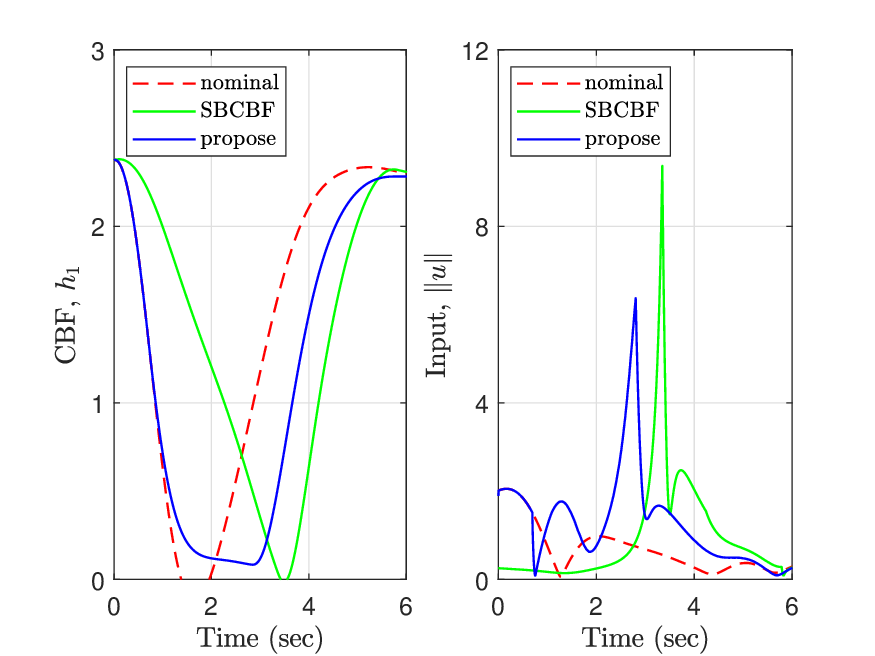}
    \caption{The evolution of the desired candidate CBF $h_1(p)$ wrt time for each controller is presented. This illustrates that, in contrast to the nominal controller and SBCBF, our proposed method guarantees $h_1(p)>0$, thereby maintaining safety $\forall t \geq 0$. Although the nominal law requires less effort than our method, it does not consider safety. Moreover, the SBCBF spends more control effort in comparison to our proposed technique, yet still fails to ensure safety.}
    \label{fig:result2}
\end{figure}

\section{Conclusions}
We propose a robust safety-critical control framework for chains of integrators subject to matched and mismatched perturbations, inspired by hyperexponential stability for nonautonomous systems. By integrating linear time-varying feedback, backstepping, and quadratic programming, the proposed method guarantees both safety and robustness while avoiding the singularities inherent in prescribed-time approaches. Rigorous theoretical guarantees are established for the double integrator case, and simulation results validate the effectiveness of the approach. This framework provides a systematic foundation for extending robust safety-critical control to higher-order and more general nonlinear systems.

\bibliographystyle{IEEEtran}
\balance
\bibliography{bib}

\end{document}